\documentclass[letterpaper,10pt,conference]{ieeeconf}  

\IEEEoverridecommandlockouts                              

\usepackage{graphics} 
\usepackage{epsfig} 
\usepackage{amsmath} 
\usepackage{amssymb}  
 
\usepackage{amsthm}
\usepackage{cite}
\usepackage{color}
\usepackage{enumerate}
\usepackage{graphicx}
\usepackage{mathtools}
\usepackage{cancel}
\usepackage{algorithm}
\usepackage{algpseudocode}
\usepackage{textgreek}
\usepackage{setspace}
\usepackage{cleveref}
\crefname{equation}{}{}
\crefname{proposition}{Proposition}{Propositions}
\usepackage{url}
\newcommand{\nspace}[1]{\mkern#1mu}
\graphicspath{{figures/}}

\newcommand{\naturals}{\mathbb{N}}
\newcommand{\real}{\mathbb{R}}

\newcommand{\realpos}{\mathbb{R}_{> 0}}

\newcommand{\map}[3]{#1:#2 \rightarrow #3}

\newcommand{\setdefb}[2]{\big\{#1 \; | \; #2\big\}}

\newcommand*{\SetSuchThat}[1][]{} 
\newcommand*{\MvertSets}{%
    \renewcommand*\SetSuchThat[1][]{%
        \mathclose{}%
        \nonscript\;##1\vert\penalty\relpenalty\nonscript\;%
        \mathopen{}%
    }%
}
\MvertSets 
\DeclarePairedDelimiterX \Set [2] {\lbrace}{\rbrace}
    {\,#1\SetSuchThat[\delimsize]#2\,}

\newcommand{\Cc}{\mathcal{C}}
\newcommand{\Dc}{\mathcal{D}}

\newcommand{\Fc}{\mathcal{F}}

\newcommand{\Hc}{\mathcal{H}}

\newcommand{\Kc}{\mathcal{K}}

\newcommand{\Uc}{\mathcal{U}}
\newcommand{\Rc}{\mathcal{R}}
\newcommand{\Sc}{\mathcal{S}}

\newcommand{\Xc}{\mathcal{X}}

\newcommand{\R}{\mathbb{R}}
\newcommand{\Tb}{T_{\rm b}}

\newcommand{\Ke}{\Kc^{\rm e}}

\newcommand{\defeq}{\triangleq}

\newtheorem{theorem}{Theorem}
\newtheorem{lemma}{Lemma}

\newtheorem{proposition}{Proposition}

\theoremstyle{definition}
\newtheorem{definition}{Definition}

\newtheorem{remark}{Remark}
\newtheorem{example}{Example}
\newtheorem{assumption}{Assumption}
\newtheorem{problem}{Problem}

\renewcommand{\bf}{\mathbf{f}} 
\newcommand{\bg}{\mathbf{g}}

\newcommand{\bk}{\mathbf{k}}

\newcommand{\bu}{\mathbf{u}}

\newcommand{\bx}{\mathbf{x}}

\newcommand{\bphi}{\boldsymbol{\varphi}}

\newcommand{\kb}{\bk_{\rm b}}

\renewcommand{\baselinestretch}{1}
\allowdisplaybreaks

\DeclareMathOperator*{\argmin}{argmin}

\title{\LARGE \textbf{Filtering for Liveness: \\ Recurrence via Backup Reachability Control Barrier Functions}}

\author{David E. J. van Wijk, Pio Ong, Aaron D. Ames, Joel W. Burdick 
\thanks{DVW and PO are with the Amazon SAF Lab, \texttt{\{vanwijk, pioong\}@amazon.com}.
This work does not relate to the authors' positions at Amazon. ADA and JWB are with the Department of Mechanical and Civil Engineering, California Institute of Technology, Pasadena, CA 91125, USA. \texttt{\{ames, jburdick\}@caltech.edu}.}
\thanks{This work was in part supported by the Technology Innovation Institute, and DARPA under the LINC program.}
}

\begin{document}
\maketitle

\begin{spacing}{0.97}
\begin{abstract}
Liveness requires a dynamical system to repeatedly accomplish a desired objective over time. In this work, we consider a particular form of liveness known as recurrence, which requires the closed-loop trajectory to regularly revisit a prescribed set, with no more than a specified amount of time between visits. Inspired by the control barrier function (CBF) safety-filter paradigm, we develop a controller-filtering framework that makes adjustments to a nominal controller only when necessary to guarantee recurrence. Unlike conventional CBF methods, which enforce set forward invariance, our framework allows the trajectory to leave the set while guaranteeing its timely return. To achieve this, our framework combines two components: a \textit{backup reachability CBF} that preserves the ability to return to the prescribed set within the allowable time, and a timing-based mechanism that invokes a reaching controller only when the nominal behavior would otherwise fail to meet the return-time requirement. The resulting framework guarantees recurrence while preserving the nominal behavior when possible, enabling control objectives that require temporary excursions from the set, such as traversal between disjoint regions. Our method is successfully utilized in a drone racing example with a recurrence constraint on beacon visibility.
\end{abstract}

\section{Introduction}

Safety is essential for the deployment of autonomous systems in real-world environments, and therefore establishing formal safety guarantees has become a central research domain.
Many approaches rely on the notion of \textit{forward invariance}, where safety is enforced by requiring the system trajectory to remain within a prescribed set for all times. These approaches include control barrier functions (CBFs) \cite{ADA-XX-JWG-PT:17}, Hamilton--Jacobi reachability \cite{SB-MC-SH-CJT:17,choi2021robust}, model predictive control (MPC) \cite{GARCIA1989335}, and reference governors~\cite{GARONE2017306,freire_DSM_26}.

In addition to safety, many autonomous systems must also satisfy \textit{liveness} objectives, which require desired events or conditions to occur over time \cite{ALPERN1985181,10066157}. In this letter, we focus on a particular type of liveness, known as \textit{recurrence}, which states that the dynamical system is allowed to leave a prescribed set as long as it returns within a pre-specified amount of time. Such requirements arise, for example, when a system must regularly visit a region where localization measurements are available. Although recurrence differs from forward invariance in allowing temporary excursions from the prescribed set, the two notions are closely related from a control design perspective. This connection suggests that tools developed for safety and forward invariance may provide a good foundation for enforcing recurrence.
\begin{figure}[t]
    \centering
    \includegraphics[width=.95\linewidth]{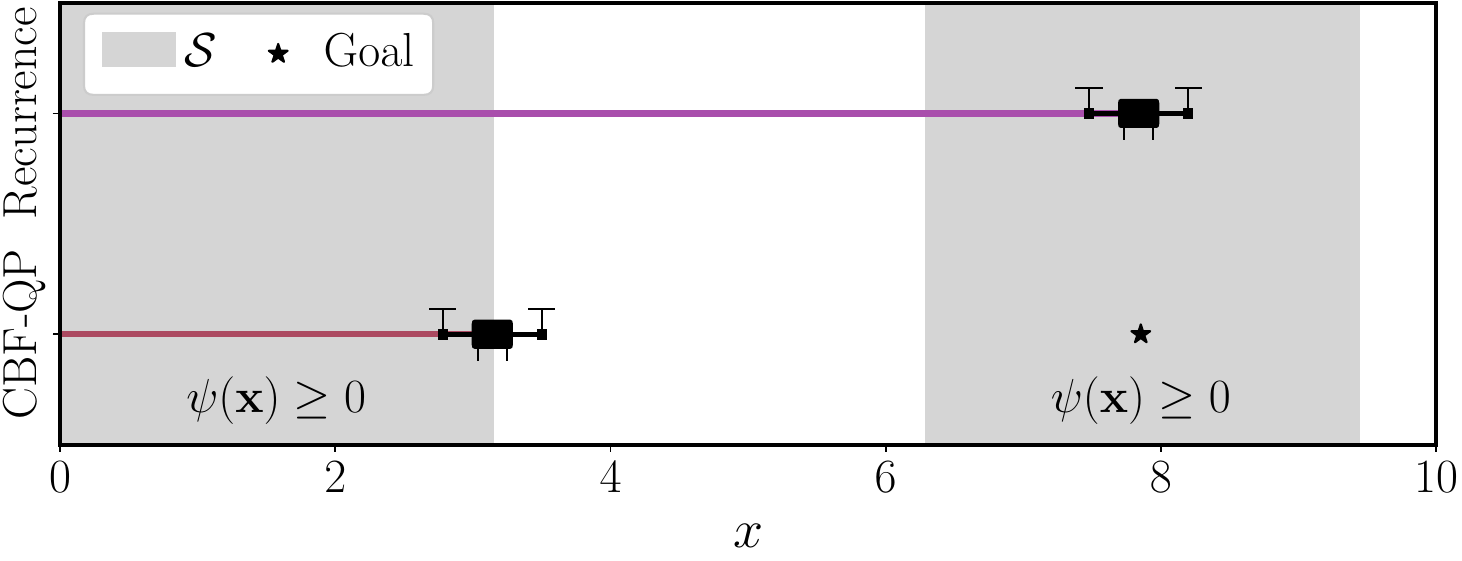}
    \vspace{-.45cm}
    \caption{{Depiction of recurrence versus invariance. Forward invariance endowed by the standard CBF-QP (red) yields forward recurrence but may be restrictive when just recurrence will suffice. In such cases, traveling between disjoint sets is allowed and the goal can be reached (purple).}}
    \label{fig:hero_single_integrator}
    \vspace{-.6cm}
\end{figure}
Our work builds on CBFs, a popular tool for certifying forward invariance and characterizing control inputs that satisfy state constraints. In particular, CBF constraints naturally lend themselves to \textit{controller filtering}, where a nominal task-oriented controller is modified only when necessary to maintain the desired constraint~\cite{ADA-SC-ME-GN-KS-PT:19}. This \textit{filtering} paradigm is particularly attractive for balancing task performance and constraint satisfaction.
Recurrence has recently been studied in the context of CBFs \cite{liu2025recurrent}
where the notion of recurrent CBFs was developed to verify recurrence of prescribed sets. The authors have also studied this notion in the context of stability \cite{10383373}, and regions of attraction \cite{9993280}. However, these approaches primarily provide tools for set verification, rather than controller filtering.
Another related notion to recurrence is 
periodic safety \cite{garg22Lcss}. Unlike recurrence, the system must reach the prescribed set at exact time instances while remaining in a safe set at all times.

Recurrence can also be viewed through the lens of signal temporal logic (STL), where bounded temporal operators can be used to specify repeated return to a prescribed set~\cite{STL_og}. While such formulations provide a natural language for expressing recurrence, enforcing it through a minimally invasive controller remains challenging.
Existing works have combined STL specifications with CBFs~\cite{LL-DVD:19,9147387}; however these works typically focus on \textit{safety} rather than \textit{recurrence}.

Drawing on connections between recurrence and reachability, the main result in this letter relies on ideas from \textit{backup control barrier functions} \cite{TG-MM-ADA-EF:18,TG-MM-AS-PN-EF-ADA:20,vanWijk_DRbCBF_24}. Backup CBFs generate safety guarantees by filtering for safe reachability to a known safe region under a backup controller. However, existing backup-CBF-related works focus specifically on forward invariance, rather than recurrence.

We build on the bCBF and recurrent CBF works to create a recurrence filtering framework. Our approach combines the notion of a \textit{backup reachability CBF} that guarantees the ability to return to the prescribed set within an allowable time, with a timing-based process that invokes a backup controller when a violation of the recurrence constraint is imminent.
We summarize our contributions below:
\begin{itemize}
    \item We establish a set-based characterization of controlled $T$-recurrence through finite-time backward reachability, replacing the trajectory-wise recurrence condition;
    \item We develop a tractable sufficient condition for controlled recurrence based on reachability under a backup controller through a novel \textit{backup reachability CBF};
    \item Using the reachability CBFs, we develop a timing-based switching mechanism for \textit{recurrence filtering} that allows the filtered nominal controller to act as long as there is a guarantee on a timely return to the prescribed set, thereby establishing forward recurrence of the set while preserving the nominal behavior whenever possible.
\end{itemize}

\section{Safety with Control Barrier Functions}\label{sec:CBF}

Consider the control affine dynamical system:
\begin{equation}\label{sys:ctrl_affine}
    \dot\bx = \bf(\bx)+\bg(\bx)\bu,
\end{equation}
with system state ${\bx\in\mathcal{X}}$ and control input ${\bu\in\Uc}$, where ${\Uc \subseteq\real^m}$ is a polytopic set. The drift dynamics ${\bf:\mathcal{X} \to \real^n}$ and the control matrix ${\bg:\mathcal{X} \to\real^{n \times m}}$ are assumed to be continuously differentiable. 
Given a controller ${\bk : \Xc \to \Uc}$, the state feedback ${\bu=\bk(\bx)}$ yields the  closed-loop system:
\begin{equation}\label{sys:closed_loop}
    \dot\bx = \bf_{\bk}(\bx) = \bf(\bx)+\bg(\bx)\bk(\bx),
\end{equation}
where a solution ${t\mapsto \bx(t)}$ is called a \textit{closed-loop trajectory}. When \eqref{sys:closed_loop} admits a unique forward-complete solution from every initial condition $\bx_0\!\in\!\Xc_0\!\subseteq\!\Xc$, we denote the corresponding solution map by ${\map{\bphi_\bk}{\real_{\geq 0}\!\times\!\Xc_0}{\Xc}}$ such that ${\bphi_\bk(t,\bx_0)=\bx(t)}$ when initialized at $\bx_0$.
In safety problems, we consider a state constraint set defined by a continuously differentiable \textit{constraint function} ${\map{\psi}{\mathcal{X}}{\real}}$:
\begin{equation}\label{eq:safety_constraint}
    \Sc \defeq \setdefb{\bx\in\mathcal{X}}{\psi(\bx)\geq 0}.
\end{equation}

Typically, safety requires the closed-loop trajectories $\bx(t)$ initialized at ${\bx_0\!\in\!\Sc}$ to stay in the set for all ${t\!\geq \!0}$. 
One can think of this as guaranteeing the forward invariance of $\Sc$.
\begin{definition}[Forward Invariance]
    A set ${\Cc \subseteq \mathcal{X}}$ is \textit{forward invariant} for the closed-loop system \eqref{sys:closed_loop} if $\bx_0 \triangleq\bx(0) \in \Cc \implies \bx(t) \in \Cc$, for all ${t \geq 0}$.
\end{definition}
For control systems, the corresponding notion is \textit{controlled invariance}, which characterizes whether the state trajectory can be kept within the set using admissible control inputs.
\begin{definition}[Controlled Invariance]
    A set ${\Cc \subseteq \mathcal{X}}$ is \textit{controlled invariant} for system \eqref{sys:ctrl_affine} if for every $\bx_0 \in \Cc$ there exists a control signal $t\mapsto\bu(t)$ satisfying $\bu(t)\in\Uc$ for all $t\geq 0$ such that the trajectory $\bx(t) \in \Cc$ for all $t\geq 0$.
\end{definition}
Especially when the control inputs are bounded, the set $\Sc$, defined by the constraint function $\psi$, is unlikely controlled invariant. In this case, practitioners often seek a \textit{subset} $\Cc_0\subseteq\Sc$ to render forward invariant instead. This set $\mathcal{C}_{0}$ is defined by a continuously differentiable function ${h : \mathcal{X} \rightarrow \R}$ as: 
\begin{align}\label{def:C0}
    \mathcal{C}_{0} \defeq \setdefb{\bx\in\mathcal{X}}{h(\bx)\geq 0} \subseteq \Sc.
\end{align}
To this end, control barrier functions (CBFs) are a tool for certifying controlled invariance of the set $\Cc_0$. 
\begin{definition}[Control Barrier Function \cite{ADA-XX-JWG-PT:17}]
    A continuously differentiable function ${h : \mathcal{X} \rightarrow \R}$ is a \textit{control barrier function (CBF)} for system~\eqref{sys:ctrl_affine} if there exists a class-$\Ke$ function $\alpha$ such that, for every ${\bx \in \mathcal{C}_0}$, there exists $\bu\in\Uc$ satisfying:
    \begin{equation}
    \label{eq:CBF}
        \dot h(\bx,\bu)\defeq L_\bf h(\bx)+L_\bg h(\bx)\bu > -\alpha(h(\bx)).
    \end{equation}
\end{definition}
The strict inequality in \eqref{eq:CBF} ensures that a control input $\bu\in\Uc$ satisfying the inequality exists not only on $\Cc_0$, but also on a neighborhood ${\Dc\supset \Cc_0}$. Given a continuous primary controller ${\bk_{\rm p}\! : \!\mathcal{X} \!\rightarrow \!\R^m}$ the following  quadratic program (QP):
\begin{align} 
    \bk_{\rm safe}(\bx) = \underset{\mathbf{u} \in \mathcal{U}}{\argmin} \mkern9mu &
    \left\Vert \mathbf{u} - \mathbf{k}_{\rm p}(\bx) \right\Vert^{2}
    \tag{CBF-QP} \label{eq:cbf-qp} \\
    \text{s.t.} \quad &
    \dot{h}(\bx, \mathbf{u}) \ge -\alpha(h(\bx)),
    \nonumber
\end{align}
is feasible on $\Dc$. Moreover, strict feasibility (Slater's condition) together with continuity of the problem data and convexity of the constraints ensures the resulting controller $\bk_{\rm safe}$ is continuous on $\Dc$, cf. \cite{PM-AA-JC:25}. The following main result for CBFs establishes forward invariance using only continuity of the closed-loop controller.
\begin{theorem}[\hspace{-0.01em}\cite{ADA-XX-JWG-PT:17}] \label{thm: cbf}
Consider the control affine system~\eqref{sys:ctrl_affine}. Let $\Dc\subseteq \Xc$ be an open neighborhood of $\Cc_0$.
Any continuous controller ${\bk:\mathcal{X} \to \mathcal{U}}$ satisfying: 
\begin{align} \label{eq: cbf_condition}
    \dot{h}(\bx, \bk(\bx)) \ge -\alpha(h(\bx)),
\end{align}
for all ${\bx \in \Dc}$ renders the set ${\mathcal{C}_0 \subseteq \Sc}$ forward invariant for the closed-loop system under state feedback $\bu=\bk(\bx)$.
\end{theorem}
By construction, the \eqref{eq:cbf-qp} satisfies the condition in \Cref{thm: cbf}, and hence, it renders $\Cc_0$ forward invariant for the closed-loop system.
Since ${\mathcal{C}_0 \subseteq \Sc}$, any trajectory initialized in $\Cc_0$ remains within the constraint set $\Sc$.
The \eqref{eq:cbf-qp} has proven very useful for many robotic and aerospace applications, as it enables the synthesis of safe controllers in a computationally tractable manner. However, in some cases, the notion of forward invariance may be overly restrictive. For example, in applications involving regions of interest or regions where localization measurements are available, it may suffice for the trajectory to return to a desired set regularly. We explore this idea in the following sections. 

\section{Recurrent Safety Constraints}

The CBF methods presented in the preliminaries focus on rendering a subset of the constraint set forward invariant, thereby requiring the closed-loop trajectories remain within the set for all future time. In some applications, however, a liveness requirement \cite{ALPERN1985181,10066157} may be more appropriate. In particular, we consider recurrence, where the closed-loop trajectory is allowed to leave a prescribed set but is required to revisit it regularly. We use the following definition.

\begin{definition}[Forward $T$-recurrence] \label{def:recur_forward}
Given a time horizon $T>0$, a set $\Cc\subseteq\Xc$ is \textit{forward $T$-recurrent} for the closed-loop system \eqref{sys:closed_loop} if, for any closed-loop trajectory $t\mapsto \bx(t)$ initialized from $\bx_0 \in \Cc$, we have for each $t\geq 0$:
\begin{align}
    \exists \, t' > t \nspace{7} {\rm with} \nspace{7} t' - t \leq T \nspace{7} {\rm s.t.} \nspace{7} \bx(t') \in \Cc.
\end{align}
\end{definition}
The corresponding control notion is \textit{controlled recurrence}, introduced in \cite{liu2025recurrent}, which requires the existence of an admissible control signal that renders the set forward recurrent.

\begin{definition}[Controlled $T$-recurrence \cite{liu2025recurrent}] \label{def:recur_controlled_inv}
Given a time horizon ${T>0}$, a set ${\Cc\subseteq\Xc}$ is \textit{controlled $T$-recurrent} for system~\eqref{sys:ctrl_affine} if, for every ${\bx_0\in\Cc}$, there exists a control signal ${t\mapsto\bu(t)}$ satisfying ${\bu(t)\in\Uc}$ for all ${t\geq 0}$ such that the resulting trajectory ${t\mapsto \bx(t)}$ satisfies, for each ${t \geq 0}$:
\begin{align}
    \exists \, t' > t \nspace{7} {\rm with} \nspace{7} t' - t \leq T \nspace{7} {\rm s.t.} \nspace{7} \bx(t') \in \Cc.
\end{align}
\end{definition}

Recurrence can be viewed as a relaxation of invariance, with invariance recovered in the limit as ${T \to 0^+}$. Thus, enforcing forward invariance when only recurrence is needed may be overly restrictive and may prevent the system from accomplishing its objectives, as seen in this simple example.
\begin{example}[Single Integrator]
Consider a one dimensional motion of a quadrotor modeled by a single integrator: 
\begin{equation} \label{sys:single_integrator}
    \dot{x} = u,
\end{equation}
with control ${u \in [-1,1]}$. The constraint set $\Sc$ is defined by the constraint function ${\psi(x) = \sin(x)}$
as in~\eqref{eq:safety_constraint}. Suppose that the set represents regions where localization measurements are available, and that the quadrotor must revisit such a region within every ${T=\pi}$ time units to maintain localization. That is, the set $\Sc$ is required to be forward $T$-recurrence. Let ${x_0 = 0}$, and suppose the objective is to reach ${x_{\rm g} = \frac{9\pi}{2}}$. 

Since forward invariance implies forward $T$-recurrence, the recurrence requirement may be satisfied by rendering $\Sc$ forward invariant, e.g., using the \eqref{eq:cbf-qp}. However, since the localization regions are disjoint, forward invariance prevents task completion, see \Cref{fig:hero_single_integrator}.
In contrast, using the notion of recurrence
permits temporary excursions from $\Sc$, allowing the trajectory to reach the goal.~\hfill $\bullet$
\end{example}
The example motivates the development of a controller filter for enforcing $T$-recurrence while preserving the nominal controller whenever possible. To this end, we seek conditions that characterize controlled recurrence in terms of sets rather than individual trajectories.
As in the CBF setting of \Cref{sec:CBF}, the set $\Sc$ itself may not be controlled $T$-recurrent. We therefore consider a subset of ${\Cc_0\subseteq\Sc}$, defined by a function $h$ as in \eqref{def:C0}.
We first recall the notion of the backward reachable set of ${\Cc_0}$. For system~\eqref{sys:ctrl_affine}, define:
\begin{align} \label{eq:BRS}
\begin{aligned}
    \mathcal{R}_{T}(\mathcal{C}_0)\!\triangleq \!\big\{ \bx_0 \!\in \!\Xc\! ~\vert~& \exists \nspace{2} \bu \in \mathcal{U}^{(0,T]} \nspace{7} {\rm s.t.}~{\rm with}~\bx(0)=\bx_0 ,\\ \nspace{7} \exists &\nspace{2} t \!\in\! (0,T], \nspace{2} \nspace{2} \bx(t) \!\in \!\mathcal{C}_0~ \big\}. \raisetag{12pt}
\end{aligned}
\end{align}
That is, the backward reachable set $\Rc_{T}(\Cc_0)$ consists of all states that can be controlled to reach $\Cc_0$ within the time horizon $T$.
Using this concept, we obtain the following set-based characterization of controlled $T$-recurrence.

\begin{theorem}[Equivalence of Controlled Recurrence]\label{thm:BRS}
Consider system~\eqref{sys:ctrl_affine} and a candidate set $\mathcal{C}_0$. The set $\mathcal{C}_0$ is controlled $T$-recurrent if and only if ${\mathcal{C}_0 \subseteq \mathcal{R}_{T}(\mathcal{C}_0)}$.
\begin{proof}
    \textbf{Sufficiency:}
    Assume ${\Cc_0\subseteq \Rc_{T}(\Cc_0)}$. Then by definition of $\Rc_{T}(\Cc_0)$, for each initial condition ${\bx^{(0)}=\bx_0\in\Cc_0}$, there exists a control signal $t\mapsto \bu^{(0)}(t)$, defined on the interval $(0,T]$, such that the corresponding trajectory initialized at $\bx^{(0)}$ satisfies $\bx(t^{(1)})\in\Cc_0$ at time $t^{(1)}\in (0,T]$. Define the state $\bx^{(1)} = \bx(t^{(1)})$ as the next initial condition. Then by induction, assuming $\bx^{(k)}\in\Cc_0$, there exists a control signal $t\mapsto \bu^{(k)}(t)$, defined on the interval $(t^{(k)},t^{(k)}+T]$, such that the corresponding trajectory initialized at $\bx(t^{(k)})=\bx^{(k)}$ satisfies $\bx(t^{(k+1)})\in\Cc_0$ at time $t^{(k+1)}\in(t^{(k)},t^{(k)}+T]$. 
    
    From this construction, the time sequence $\{t^{(k)}\}_{k=0}^\infty$ has two possibilities. Case $\lim_{k\to \infty} t^{(k)} = \infty$: we can define the signal $t\mapsto \bu(t)$ piecewise as:
    \begin{align*}
    \bu(t)= \bu^{(k)}(t)~\forall t\in(t^{(k)},t^{(k+1)}]
    \end{align*}
    for each $k\in\naturals$. This control signal $t\mapsto \bu(t)$ satisfies $\bu(t)\in\Uc$ for all $t\geq 0$ and renders $\bx(t^{(k)}) = \bx^{(k)}\in\Cc_0$. Noting that $t^{(k+1)}-t^{(k)}\leq T$ for all $k\in\naturals$, the set $\Cc_0$ is controlled $T$-recurrent. Case $\lim_{k\to \infty} t^{(k)} = t^\star\in\realpos$: we have that $\bx(t^*)\in\Cc_0$ from the set being closed. Therefore, we can continue to define $\bu(t)$ for $t>t^\star$ with the same process as above, until it is defined for all $t\geq 0$, and conclude similarly that the set $\Cc_0$ is controlled $T$-recurrent.

    \textbf{Necessity:}
    Assume $\Cc_0$ is controlled $T$-recurrent. Then by definition, given any $\bx_0\in\Cc_0$, there exists a control signal $t\mapsto \bu(t)$, such that the corresponding trajectory satisfies $\bx(t')\in \Cc_0$ at time $t'>0$ with $t'\leq T$. Hence, $\bx\in\Rc_{T}(\Cc_0)$, implying $\Cc_0 \subseteq\Rc_{T}(\Cc_0)$ as desired.
\end{proof}
\end{theorem}

While \Cref{thm:BRS} gives necessary and sufficient conditions for verifying that a candidate set $\mathcal{C}_0$ is controlled $T$-recurrent, computing the backward reachable set may be challenging for high-dimensional nonlinear systems, especially in the presence of input constraints. Although approaches to compute these sets have been well-studied in the literature, their scalability remains an active area of research.

To make the problem more tractable, we develop sufficient conditions for controlled $T$-recurrence by considering reachability under a \textit{single feedback controller}, as opposed to all possible sequences of control signals. In particular, consider a backup controller ${\bk_{\rm b}: \mathcal{X} \to \mathcal{U}}$. The feedback law $\bu=\bk_{\rm b}(\bx)$ defines a particular admissible control signal $\bu(t)$ in the definition~\eqref{eq:BRS} of the backward reachable set. Suppose that the resulting closed-loop system admits a unique solution over the horizon $[0,T]$, then
the flow is the solution to the following ordinary differential equation (ODE):
\begin{align} \label{eq:dyn_fk}
    \tfrac{\partial}{\partial \tau}\boldsymbol{\varphi}_{\bk_{\rm b}}\!(\tau,\bx)  = \bf_{\bk_{\rm b}}\!(\boldsymbol{\varphi}_{\bk_{\rm b}}\!(\tau,\bx)), \quad    \boldsymbol{\varphi}_{\bk_{\rm b}}\!(0,\bx) = \bx,
\end{align}
where 
${\bf_{\bk_{\rm b}}\!(\bx) = \bf(\bx) + \bg(\bx)\bk_{\rm b}(\bx)}$.
Using this flow, we construct ${\mathcal{C}_{T}(\mathcal{C}_0)}$ consisting of states $\bx$ from which $\mathcal{C}_0$ may be reached at a finite time ${T \!>\! 0}$ under the controller $\bk_{\rm b}$.
\begin{align} \label{eq:reachable_set_Tb}
\begin{aligned}
    \mathcal{C}_{T}(\mathcal{C}_0) \triangleq \{ \bx \in \mathcal{X} : \boldsymbol{\varphi}_{\bk_{\rm b}}\!(T,\bx) \in \mathcal{C}_0 \} \subseteq \mathcal{R}_{T}(\mathcal{C}_0), \\
    = \{ \bx \in \mathcal{X} : h_{\rm r}(\bx) \triangleq h\big(\boldsymbol{\varphi}_{\bk_{\rm b}}\!(T,\bx)\big) \geq 0 \}.
\end{aligned}
\end{align}
Note that as opposed to $\mathcal{R}_{T}(\mathcal{C}_0)$, $\mathcal{C}_{T}(\mathcal{C}_0)$ specifies that the system under $\bk_{\rm b}$ must reach $\mathcal{C}_0$  at \textit{exactly} time $T$. A more general construction that captures reachability \textit{within} time $T$ can also be developed; we briefly discuss this extension later in \Cref{rmk:combinatorial}, and defer its full development due to space limitations. However, even without this additional condition, the states in $\Cc_T(\Cc_0)$ can be controlled to reach $\Cc_0$ using the feedback law, and therefore, the set $\Cc_T(\Cc_0)$ is an inner approximation of the backward reachable set $\Rc_{T}(\Cc_0)$. Combining this inclusion with \Cref{thm:BRS} immediately yields the following sufficient condition for controlled $T$-recurrence.
\begin{lemma}[Sufficient Condition for Controlled Recurrence] \label{lemma:sufficient_CT}
Consider the closed-loop system~\eqref{sys:closed_loop} under the feedback controller ${\bk_{\rm b}: \mathcal{X} \!\to\! \mathcal{U}}$. If $\mathcal{C}_0 \subseteq \mathcal{C}_{T}(\mathcal{C}_0)$, then
$\mathcal{C}_0$ is controlled $T$-recurrent for~\eqref{sys:ctrl_affine}.
\begin{proof} 
    Under the assumption, we directly deduce the relation $\Cc_0\subseteq\Cc_{T}(\Cc_0) \subseteq \Rc_{T}(\Cc_0)$, which by \Cref{thm:BRS} is equivalent to $\Cc_0$ being controlled $T$-recurrent.
\end{proof}
\end{lemma}

By restricting our attention to a fixed feedback controller $\bk_{\rm b}$, we can avoid computing the full backward reachable set $\Rc_T(\Cc_0)$, and \Cref{lemma:sufficient_CT} reduces the verification of controlled $T$-recurrence to checking the set inclusion $\mathcal{C}_0 \subseteq \mathcal{C}_{T}(\mathcal{C}_0)$. The following proposition shows that this inclusion is automatically satisfied when $\bk_{\rm b}$ renders $\Cc_0$ forward invariant.

\begin{proposition} \label{prop:backupController}
Consider the closed-loop system~\eqref{sys:closed_loop} under the feedback controller ${\bk_{\rm b} : \mathcal{X} \rightarrow\mathcal{U}}$. If ${\mathcal{C}_0}$ is forward invariant, then, for any $T>0$, ${\mathcal{C}_0 \subseteq \mathcal{C}_{T}(\mathcal{C}_0)}$. As a consequence, $\Cc_0$ is controlled $T$-recurrent for~\eqref{sys:ctrl_affine}.
\begin{proof}
For any ${\bx\in\Cc_0}$, forward invariance of $\Cc_0$ implies ${\boldsymbol{\varphi}_{\bk_{\rm b}}\!(t,\bx)\in\Cc_0}$ for all time ${t\geq 0}$. In particular at time $T$, we have ${\boldsymbol{\varphi}_{\bk_{\rm b}}\!(T,\bx)\in\Cc_0}$, which by definition implies $\bx\in\Cc_T(\Cc_0)$. Hence, ${\mathcal{C}_0 \subseteq \mathcal{C}_{T}(\mathcal{C}_0)}$ as desired. The $T$-recurrence statement follows from \Cref{lemma:sufficient_CT}.
\end{proof}
\end{proposition}

\begin{remark}[Connection to Backup CBFs]
   A controller satisfying \Cref{prop:backupController} can be constructed via Lyapunov-based methods, where $\Cc_0$ is chosen from a suitable Lyapunov forward invariant sub-level set; see, e.g., \cite{LG-AKK-TGM:25,vanWijkGeneralized26}. For this reason, the existence of such a \textit{backup controller and set pair} is usually assumed in the \textit{backup CBF} literature 
   \cite{gurriet_online_2018,TG-MM-AS-PN-EF-ADA:20}.
\end{remark}

\Cref{prop:backupController} provides a simple means of constructing $\bk_{\rm b}$ satisfying ${\Cc_0\subseteq \Cc_T(\Cc_0)}$. Moreover, since the result holds for any horizon ${T>0}$, forward invariance ensures the exact-time requirement used to define $\Cc_T(\Cc_0)$ also captures reachability within the horizon $T$. Nevertheless, forward invariance of $\Cc_0$ under $\bk_{\rm b}$ is only a sufficient condition and is not required for our subsequent development. We therefore impose only the following weaker condition. 
\begin{assumption}\label{ass:controller_k}
The controller ${\bk_{\rm b}: \mathcal{X} \to \mathcal{U}}$ is designed such that ${\mathcal{C}_0 \subseteq \mathcal{C}_{T}(\mathcal{C}_0)}$.
\end{assumption}

We next consider the \textit{synthesis} of controllers that render the safe set $\Cc_0$ forward $T$-recurrent. The results developed in this section already provide means of constructing controllers that achieve this objective; for example, \Cref{prop:backupController} shows that a controller rendering $\Cc_0$ forward invariant is sufficient. However, such a controller may be unnecessarily restrictive if applied at all times. Instead,
we are interested in a \textit{controller-filtering} framework that preserves the behavior of a nominal controller whenever possible and intervenes only when necessary to guarantee forward $T$-recurrence.

\begin{problem} \label{prob:Main}
Given a primary controller ${\bk_{\rm p} : \mathcal{X} \rightarrow \R^m}$, design a controller ${\bk: \mathcal{X} \to \mathcal{U}}$, ${\bu = \bk(\bx)}$ that minimally modifies $\bk_{\rm p}$ while rendering $\Cc_0$ forward $T$-recurrent for \eqref{sys:closed_loop}.
\end{problem}

\section{Recurrent Safety Filtering} \label{sec:MainResults}
Our controller-filtering framework to address~\Cref{prob:Main} consists of two components. First, we develop a backup reachability CBF that renders $\Cc_{\Tb}(\Cc_0)$ forward invariant, thereby preserving the ability to return to $\Cc_0$ within a return horizon ${\Tb\in(0,T)}$. Second, we introduce a timing-based switching mechanism that allows the filtered primary controller to govern the system for up to ${T-\Tb}$ time units after leaving $\Cc_0$. If the trajectory does not return within this interval, the controller switches to $\bk_{\rm b}$, guaranteeing the return to $\Cc_0$ within the remaining time for recurrence.

\subsection{Backup Reachability Control Barrier Functions}
We first address the reachability-preservation component of the filtering framework. In particular, we seek to render $\Cc_{\Tb}(\Cc_0)$ forward invariant so that the ability to return to $\Cc_0$ is maintained during nominal operation. 

Since $\Cc_{\Tb}(\Cc_0)$ is a zero superlevel set of the function $h_{\rm r}$, forward invariance of $\Cc_{\Tb}(\Cc_0)$ can be enforced using the standard CBF framework whenever $h_{\rm r}$ is itself a CBF for system~\eqref{sys:ctrl_affine}. In addition to the forward invariance guarantees from \Cref{thm: cbf}, 
membership in $\mathcal{C}_{\Tb}(\mathcal{C}_0)$ ensures that the state can return to $\Cc_0$ under the backup controller $\bk_{\rm b}$. For this reason, we refer to $h_{\rm r}$ as a \textit{backup reachability CBF.} The following theorem formalizes this reachability guarantee.

\begin{theorem} \label{thm:control_recurrence}
Consider the control-affine system \eqref{sys:ctrl_affine}. Given a continuously differentiable backup controller ${\bk_{\rm b} : \mathcal{X} \to \mathcal{U}}$ satisfying \Cref{ass:controller_k} and a return horizon ${\Tb \in (0,T)}$, define ${\mathcal{C}_{\Tb}(\mathcal{C}_0)}$ and $h_{\rm r}$ as in \eqref{eq:reachable_set_Tb}. Define the sensitivity matrix  ${\mathbf{\Phi}(\tau,\bx) \triangleq \partial \boldsymbol{\varphi}_{\bk_{\rm b}}\!(\tau,\bx)/\partial \bx}$. Any continuous controller $\map{\bk}{\Xc}{\Uc}$ satisfying, for some class-$\Ke$ function $\alpha$:
\begin{multline} \label{eq:recur_CBF_constraint}
    \nspace{-20}\nabla h(\boldsymbol{\varphi}_{\bk_{\rm b}}\!(\Tb,\bx)) \mathbf{\Phi}(\Tb,\bx) \bf_\bk(\bx) \geq -\alpha\big({h}(\boldsymbol{\varphi}_{\bk_{\rm b}}\!(\Tb,\bx))\big), \nspace{-10}
    \raisetag{6pt}
\end{multline}
for all $\bx\in \Dc$ where $\Dc$ is an open neighborhood of $\Cc_{\Tb}(\Cc_0)$, renders the set $\Cc_{\Tb}(\Cc_0)$ forward invariant for the closed-loop system~\eqref{sys:closed_loop}. Thus, for any closed-loop trajectory $t\mapsto \bx(t)$ initialized in $\Cc_{\Tb}(\Cc_0)$, we have ${\boldsymbol{\varphi}_{\bk_{\rm b}}\!(\Tb,\bx(t))\in\Cc_0}$ for all ${t\geq 0}$, and thus, every state along the trajectory admits a return to $\Cc_0$ under the backup controller $\bk_{\rm b}$ in $\Tb$ time units.
\begin{proof}
    Noting that inequality \eqref{eq:recur_CBF_constraint} is exactly the CBF inequality~\eqref{eq: cbf_condition} for $h_r$, by applying the chain rule, forward invariance follows directly from \Cref{thm: cbf}. The return trajectory guarantee follows from the definition of set $\mathcal{C}_{\Tb}(\mathcal{C}_0)$.
\end{proof}
\end{theorem}
\Cref{thm:control_recurrence} expresses the CBF condition for the backup reachability CBF $h_{\rm r}$ in terms of the sensitivity matrix $\mathbf{\Phi}(\Tb,\bx)$. This matrix can be computed alongside the backup trajectory by integrating the following ODE:
\begin{align}
    \tfrac{\partial}{\partial\tau}\mathbf{\Phi}(\tau,\bx) = \tfrac{\partial \bf_{\bk_{\rm b}}}{\partial \bx}(\boldsymbol{\varphi}_{\bk_{\rm b}}\!(\tau,\bx))\mathbf{\Phi}(\tau,\bx), \nspace{10} \mathbf{\Phi}(0,\bx) = \mathbf{I}.
\end{align}

As with standard CBFs, when $h_{\rm r}$ is a CBF, a controller satisfying \eqref{eq:recur_CBF_constraint} can be obtained through a CBF-QP:
\begin{align} 
     \bk_{\rm QP}& (\bx) = \underset{\mathbf{u} \in \mathcal{U}}{\argmin} \mkern9mu 
    \left\Vert \mathbf{u} - \mathbf{k}_{\rm p}(\bx)  \right\Vert^{2} 
    \tag{brCBF-QP} \label{eq:rcbf-qp} \\
    & \textup{s.t.} \quad \textup{\eqref{eq:recur_CBF_constraint}}.
    \nonumber
\end{align}

\subsection{Switching-Based Recurrent Controller}

We complete the filtering framework by introducing the timing-based switching mechanism that leverages the backup reachability guarantee from the previous section to enforce forward $T$-recurrence.
Consider the clock variable $\vartheta$, initialized at $\vartheta(0)=0$, with hybrid dynamics described by: 
\begin{align}\label{sys:hybrid}
    \mathcal{H} : 
    \begin{cases}
        \dot{\vartheta} = 1, \nspace{24} (\bx,\vartheta) \in \mathcal{F} , \\
        \vartheta^{+} = 0,\nspace{12}(\bx,\vartheta) \in \mathcal{J}, 
    \end{cases}
\end{align}
where ${\mathcal{F} \triangleq \{(\bx,\vartheta) \in \mathcal{X}\times [0,\infty) :  h(\bx) \leq 0 \}}$ and $\mathcal{J} \triangleq \{(\bx,\vartheta) \in \mathcal{X}\times [0,\infty) :  h(\bx) \geq 0 \}$. These hybrid dynamics capture how long the constraint $h$ has been negative. Based on this, we may design a controller that utilizes the filtered output $\bk_{\rm QP}$ whenever ${\vartheta < T - \Tb}$, and otherwise switches to the controller $\bk_{\rm b}$ that is guaranteed to bring \eqref{sys:ctrl_affine} to $\mathcal{C}_{0}$.

\begin{theorem}\label{thm:switching_controller}
Under the setting of \Cref{thm:control_recurrence} and the clock dynamics $\Hc$ defined in~\eqref{sys:hybrid},
the switched controller ${\bk_{\rm s} : \mathcal{X} \times \R_{\geq 0} \to \mathcal{U}}$ given by:
\begin{align} \label{eq:switching_controller}
    \bk_{\rm s}(\bx,\vartheta) = 
    \begin{cases}
        \bk(\bx) & \text{if } \vartheta < T - \Tb,\\
        \bk_{\rm b}(\bx) & \text{if } \vartheta \geq T - \Tb,
    \end{cases}
\end{align}
renders $\mathcal{C}_0$ forward $T$-recurrent for the closed-loop system \eqref{sys:closed_loop} under the feedback $\bu=\bk_{\rm s}(\bx,\vartheta)$.
\begin{proof}
    Assume the set $\Cc_0$ is not forward $T$-recurrent. Then there exists a time ${t^*\!\geq\! 0}$ such that for ${t\!\in\!(t^*,t^*\!+\!T)}$, we have ${\bx(t) \notin \Sc}$, as well as ${(\bx,\vartheta)\in\Fc}$. Hence, ${\dot\vartheta(t) = 1}$ over the entire duration, implying that ${\vartheta(t) \!\geq\! \vartheta(t^*)\!+\!(t\!-\!t^*)}$. Noting that ${\vartheta(t^*) \geq 0}$ by definition, we can deduce that ${\vartheta(t^*\!+\!T\!-\!\Tb)\!\geq\! T\!-\!\Tb}$. As such, the switched controller enters the backup mode ${\bk_{\rm s}(\bx,\vartheta)\!=\!\kb(\bx)}$ no later than time ${t_{\rm switch}\!=\!t^*\!+\!T\!-\!\Tb}$. Because for all $t \!\in \!(t^*, t^*\!+\!t_{\rm switch})$, ${\bk_{\rm s}(\bx,\vartheta)\!=\!\bk(\bx)}$, \Cref{thm:control_recurrence} guarantees that ${\boldsymbol{\varphi}_{\bk_{\rm b}}\!(\Tb,\bx(t))\!\in\!\Cc_0}$  and therefore ${\bx(t^*\!+\!t_{\rm switch}\!+\!T_{\rm b})\!\in\!\Cc_0}$. This is a contradiction since ${t^*\!+\!t_{\rm switch}\!+\!\Tb \!\leq\! t^*\!+\!\Tb}$, and hence, $\Cc_0$ is forward $T$-recurrent.
\end{proof}
\end{theorem}

\Cref{thm:switching_controller} completes the recurrence-filtering framework, providing a solution to~\Cref{prob:Main}. The controller $\bk$ preserves backup reachability while minimally modifying the nominal controller, and $\bk_{\rm b}$ is invoked only when the remaining time reaches the guaranteed return horizon $\Tb$. Together these mechanisms turn backup reachability into a liveness filter.

\begin{remark}[Generalization to Recurrence of $\Sc$]
    The switching mechanism can alternatively track excursions from the larger set $\Sc$ by defining the flow and jump sets, $\Fc$ and $\mathcal{J}$, in~\eqref{sys:hybrid} using the function $\psi$ rather than $h$. In this case, the filter guarantees all trajectories initialized in $\Cc_{\Tb}(\Cc_0)$ either return to the set $\Sc$ naturally or return to ${\Cc_0\subseteq\Sc}$ via switching, within time $T$. In the case where ${\Sc\subseteq\Cc_{\Tb}(\Cc_0)}$, the set $\Sc$ is rendered forward $T$-recurrent.
\end{remark}

\begin{remark}[Combining Multiple Recurrent Sets]\label{rmk:combinatorial}
    In the case where multiple $\Cc_{\Tb}(\Cc_{0,j})$ sets with associated backup reachability CBFs $h_{{\rm r} ,j}$'s are available, it is possible to render the union of these sets forward invariant by applying the maximum operator over $h_{{\rm r},j}$ and using the combinatorial CBF framework~\cite{PO-HL-TGM-DP-ADA:25-csl} to render the union of $\Cc_{0,j}$ forward $T$-recurrent. Similarly, backup reachability CBFs constructed using different return horizons may be combined this way to obtain an approximation of within-horizon reachability, rather than reachability at a single return time.
    The theoretical details are omitted for brevity, but these ideas will be successfully utilized in \Cref{ex:planar_quadrotor}.
\end{remark}

\section{Quadrotor Visibility Example}
\label{ex:planar_quadrotor}
\begin{figure}[h]
    \centering
    \includegraphics[width=.78\linewidth]{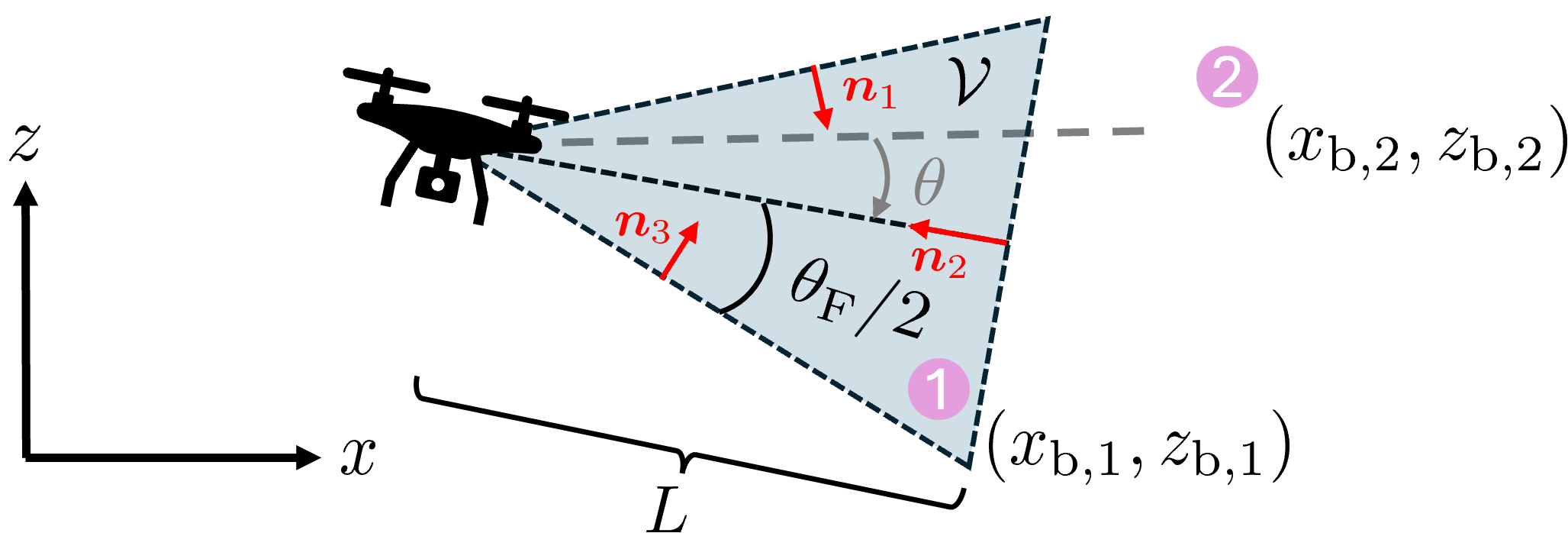}
    \vspace{-.1cm}
    \caption{Visibility cone for planar quadrotor in \Cref{ex:planar_quadrotor}. The shaded blue region $\mathcal{V}$ represents all positions that are viewable by \eqref{eq:quadrotor}.}
    \label{fig:droneCamera}
    \vspace{-.4cm}
\end{figure}
\begin{figure*}[t] 
    \centering
    \begin{minipage}{0.5\textwidth}
        \centering
        \includegraphics[width=\linewidth]{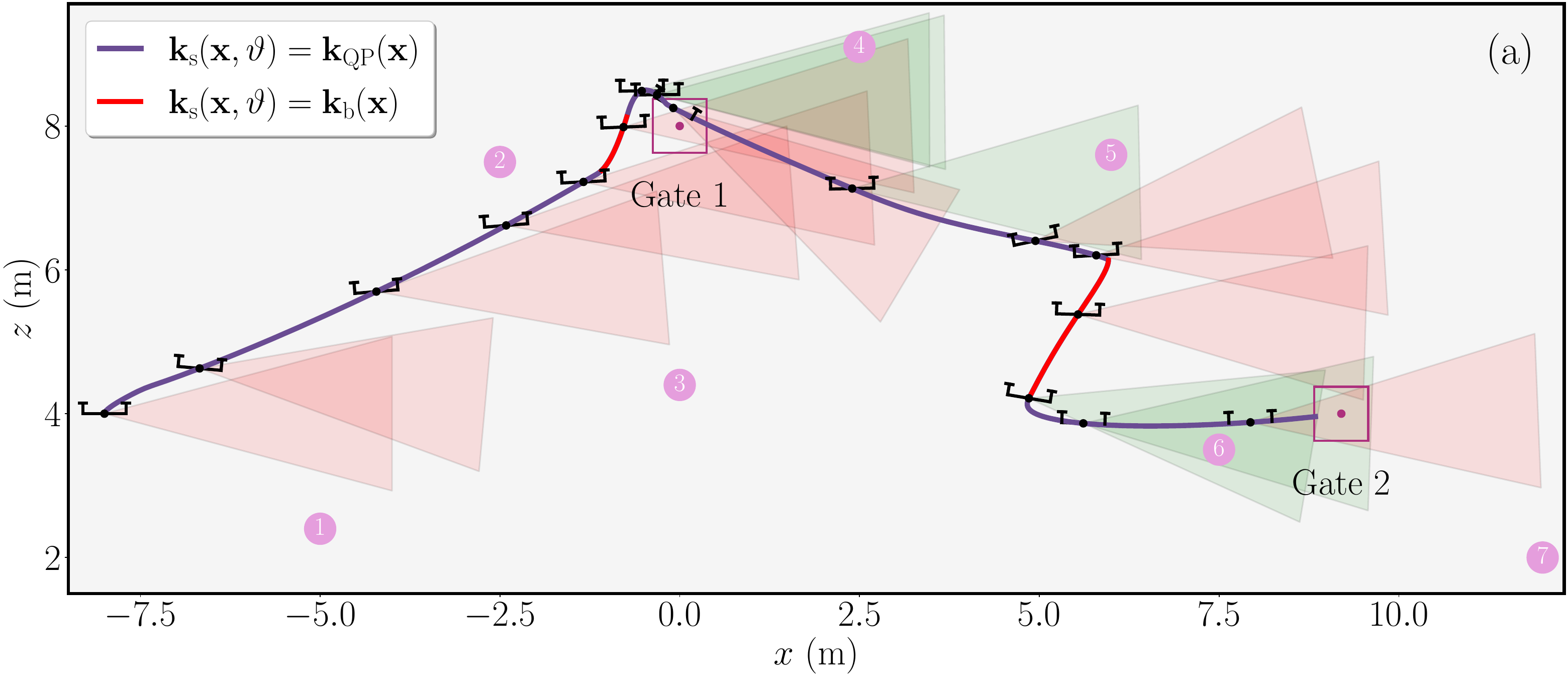}
    \end{minipage}
    \begin{minipage}{0.46\textwidth}
        \centering
        \includegraphics[width=\linewidth]{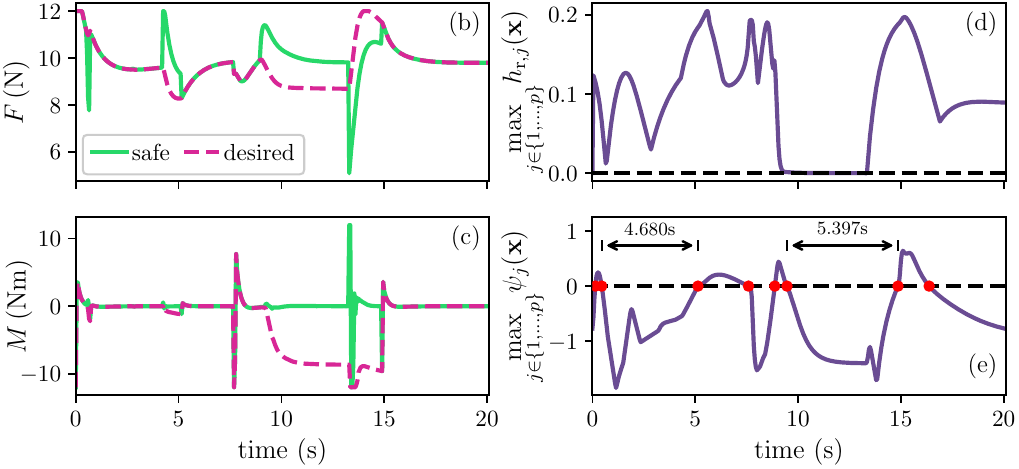}
    \end{minipage}
    \vspace{-.14cm}
    \caption{Simulation results of the drone racing problem with visibility recurrence.
    Panel (\textbf{a}) plots the trajectory of the quadrotor (purple and red) using controller \eqref{eq:switching_controller}, with ${p = 7}$ beacons (pink circles). Green and red cones indicate that a beacon is or is not in the drone's FOV, respectively. Our approach guarantees the forward $T$-recurrence for \eqref{eq:quadrotor} with ${T = 6 \,{\rm s}}$ while minimizing deviations from the primary controller whenever possible (\textbf{b},$\nspace{2}$\textbf{c}), which enables the drone to visit both gates (\textbf{a}). The value of the function describing the disjoint set $\Sc$ never drops below zero for longer than $T$ seconds, indicating forward $T$-recurrence (\textbf{e}). The reachability CBF is always kept positive (\textbf{d}), indicating that a viable backup plan always exists.}
    \vspace{-.35cm}
    \label{fig:droneMain}
    \vskip -0.1 true in
\end{figure*}
Consider a planar version of an autonomous racing quadrotor with mass ${m}$ and principal moment of inertia ${J}$:
\begin{align}
\begin{aligned}
    \ddot{x} \!=\! F \sin(\theta)/m, \nspace{5}
    \ddot{z} \!=\! F \cos(\theta)/m \!-\! g_{\rm D}, \nspace{5}
    \ddot{\theta} \!=\! -M/J,
    \label{eq:quadrotor}
\end{aligned}
\end{align}
with state ${\bx = [x \nspace{5}z\nspace{5} \theta\nspace{5} \dot{x}\nspace{5} \dot{z}\nspace{5} \dot{\theta}]^\top \!\in\! \mathcal{X}\! \subset \!\mathbb{R}^2\! \times \!\mathbb{S}^1\! \times \!\mathbb{R}^3}$.
States $x$ and $z$ represent the horizontal and vertical drone positions,
and $\theta$ is its pitch angle.  The inputs are the thrust
${F \in [0,F_{\rm max}]}$ and moment ${M \in [-M_{\rm max}, M_{\rm max}]}$. Above, ${g_{\rm D}}$ is the gravitational constant.

Consider a scenario in which the quadrotor must traverse a set of racing gates while using its front-facing camera to occasionally localize itself using a number of beacons. There are $p$ beacons at ${\mathbf{p}_{{\rm b},j} = [x_{{\rm b},j} \nspace{7} z_{{\rm b},j}]^\top \in \R^2}$, for all ${j \in \{1,\dots,p\}}$. The camera is affixed to the body with field of view given by ${\theta_{\rm F}}$ (see \Cref{fig:droneCamera}) and a fixed focal length ${L}$.
Therefore, the ${j^{\rm th}}$ beacon is in view if ${\bx \in \mathcal{S}_{j}}$ where
\begin{align}
\mathcal{S}_{j} \!=\! \bigg \{ \bx \in \mathcal{X} :  \psi_{j}(\bx) \!=\! {-\frac{1}{\kappa} \ln \Big(\nspace{-2} \sum_{i \in I} {\rm e}^{-\kappa \beta_i(\bx)}\Big)} \geq 0 \bigg \},
\end{align}
with smoothing parameter ${\kappa > 0}$ and ${I = \{1,2,\nspace{-2}3\}}$. Here, $\beta_i$ are half-space constraints of the form:
\begin{align}
    \beta_i(\bx) = \mathbf{n}_i\cdot(\mathbf{p}_{{\rm b},j} - (\mathbf{p} + \mathbf{p}^{\rm o}_i)),
\end{align}
where ${\mathbf{p} = [x \nspace{7} z]^\top}$, ${\mathbf{p}^{\rm o}_i \in \R^2}$ is a state-dependent offset, and ${\mathbf{n}_i\in \R^2}$, ${\mathbf{n}_i \neq \mathbf{0}}$ is a normal vector. 

In this scenario, we consider the forward $T$-recurrence of ${\Sc  \!=\! \{ \bx\! \in\! \mathcal{X} \!:\! \max_{j \in 
\{1,...,p\}} \psi_j (\bx) \!\geq\! 0\} \!=\! \bigcup_{j=1}^{p} \Sc_j}$, that is, the quadrotor should never go longer than $T$ seconds without viewing a beacon such that the localization uncertainty never becomes too large. We obtain $p$ subsets of $\Sc$ with ${\mathcal{C}_{0,j}}$ for ${j \in \{1,\dots,p\}}$, with corresponding controllers, $\bk_{{\rm b},j}$, that each render ${\mathcal{C}_{0,j}}$ forward invariant. This can be done as in \cite[Ex. 2]{vanWijkGeneralized26} by using a sub-level set of a Lyapunov function centered around an equilibrium point of \eqref{eq:quadrotor}, placed directly in front of a beacon. To guarantee the reachability of the union of these subsets we utilize combinatorial CBFs and the maximum operator (see \Cref{rmk:combinatorial}).

\Cref{fig:droneMain} plots the results using our approach for the complex drone visibility problem\footnote{Here we use ${T \!=\! 6} \nspace{2}{\rm s}$, ${\Tb \!= \!2.2 \nspace{2}{\rm s}}$, ${L \!= \!4 \nspace{2}{\rm m}}$, ${\theta_{\rm F}\! = \!30 ^\circ}$, ${g_{\rm D} \!=\! 9.81 \nspace{2}{\rm m/s^2}}$, ${J \!=\! {\rm 0.25 \nspace{2} kg \nspace{1} m^2}}$, ${m \!=\! {\rm 1 \nspace{2} kg}}$, ${F_{\rm max} \!=\! 12 \nspace{2} {\rm N}}$, ${M_{\rm max} \!=\! 12 \nspace{2} {\rm Nm}}$.}. The quadrotor can successfully traverse both gates whilst ensuring that a beacon is never out of view for longer than $T$ seconds. In panel (a), the purple trajectory segments show where the outputs of the \eqref{eq:rcbf-qp} are used, which ensures the reachability of at least one of the sets ${\mathcal{C}_{0,j}}$. As per \Cref{thm:switching_controller}, when the value of the clock $\vartheta$ exceeds ${T - \Tb}$, system \eqref{eq:quadrotor} will execute one of the backup controllers $\bk_{{\rm b},j}$ in order to view a beacon before time runs out. This occurs twice during the simulation, both indicated by the red trajectory segment. The first instance occurs near Gate 1, where the quadrotor increases its altitude and stabilizes its orientation to view beacon 4. The second instance occurs near the Gate 2, where the drone reduces its altitude to view beacon 6, and then completes the task by navigating through the gate. These results demonstrate successful task completion along with satisfaction of the liveness condition.

\section{Conclusion}

We presented a novel approach to guarantee the forward recurrence of nonlinear systems using a minimally invasive, filtering-based approach. Our framework combines the notion of a backup reachability CBF that ensures the existence of a controller that satisfies recurrence, and a timing-based mechanism that switches to the reaching controller when necessary. Our approach guarantees forward recurrence while 
minimizing the deviation from the nominal controller when allowable, enabling behavior that requires the system to exit the prescribed set, such as traveling between disjoint sets.
\end{spacing}

\vskip - 4mm
\begin{spacing}{0.936}
\bibliographystyle{ieeetr}
\bibliography{
    bib/alias,
    bib/PO,
    bib/main-Pio,
    bib/backup
}
\end{spacing}
\end{document}